\documentclass{elsarticle}

\usepackage{amssymb,amsmath,amsthm}
\usepackage{xspace}
\usepackage{longtable}

\theoremstyle{plain}
 \newtheorem{theorem}{Theorem}
 \newtheorem{lemma}[theorem]{Lemma}
 \newtheorem{proposition}[theorem]{Proposition}

\theoremstyle{definition}

\theoremstyle{remark}

\theoremstyle{remark}

\begin{document}

\begin{frontmatter}

\title{On the number of Dejean words over alphabets of 5, 6, 7, 8, 9 and 10 letters}

\author[msu]{Roman Kolpakov}
\ead{foroman@mail.ru}

\author[lbr]{Micha\"el Rao}
\ead{rao@labri.fr}

\address[msu]{Moscow State University,
Leninskie Gory, 119992 Moscow, Russia}

\address[lbr]{LaBRI, Universit\'e Bordeaux 1,
351 cours de la lib\'eration, 33405 Talence, France}

\begin{abstract}
We give lower bounds on the growth rate of Dejean words, \emph{i.e.} minimally
repetitive words, over a $k$-letter alphabet, for $5\le k \le 10$. Put together 
with the known upper bounds, we estimate these growth rates with the precision 
of $0.005$. As an consequence, we establish the exponential growth of the number 
of Dejean words over a $k$-letter alphabet, for $5\le k \le 10$.
\end{abstract}

\end{frontmatter}

\section{Introduction}

Let $w=a_1\cdots a_n$ be a word over an alphabet~$\Sigma$. The number~$n$
is called the {\it length} of~$w$ and is denoted by $|w|$. The symbol
$a_i$ of~$w$ is denoted by $w[i]$. A word $a_i\cdots a_j$, where $1\le 
i\le j\le n$, is called a {\it factor} of~$w$ and is denoted by $w[i:j]$. 
For any $i=1,\ldots,n$ the factor $w[1:i]$ ($w[i:n]$) is called a {\it prefix} 
(a {\it suffix}) of~$w$. A positive integer $p$ is called  a {\it period} of~$w$ 
if $a_i=a_{i+p}$ for each $i=1,\ldots ,n-p$. If $p$ is the minimal period 
of~$w$, the ratio $e(w)=n/p$ is called the {\it exponent} of~$w$. Two words 
$w',w''$ over~$\Sigma$ are called {\it isomorphic} if $|w'|=|w''|$ and 
there exists a bijection $\sigma:\Sigma\longrightarrow\Sigma$ such that 
$w''[i]=\sigma (w'[i])$, $i=1,\ldots,|w'|$. By ${\cal K}(w)$ we will denote
the set of all words over~$\Sigma$ which are isomorphic to the word~$w$.
We also denote by $|A|$ the number of elements of a finite set~$A$.
Let $|\Sigma|=k$. It is easy to note that $|{\cal K}(w)|=k!$ if $w$
contains at least $k-1$ different symbols of~$\Sigma$.

Let $W$ be an arbitrary set of words. This set is called {\it factorial}
if for any word~$w$ from~$W$ all factors of~$w$ are also contained in~$W$.
We denote by $W(n)$ the subset of~$W$ consisting of all words of length~$n$. 
If $W$ is factorial then it is not difficult to show (see, e.g.,~\cite{Brink, 
BaElGr}) that there exists the limit $\lim_{n\to\infty} \sqrt[n]{|W(n)|}$
which is called the {\it growth rate} of words from~$W$. For any words~$u, v$ 
we denote by $W^{(v)}(n)$ the set of all words from $W(n)$ which contain~$v$ 
as a suffix, and by $W^{(u, v)}(n)$ the set of all words from $W(n)$ which 
contain~$v$ as a suffix and~$u$ as a prefix.

One can mean by a repetition any word of exponent greater than~1.
The best known example of repetitions is a {\it square}; that is,
a word of the form $uu$, where $u$ is an arbitrary nonempty word.
Avoiding ambiguity\footnote{Note that the period of a square is not
necessarily the minimal period of this word.}, by the {\it period} 
of the square $uu$ we mean the length of~$u$. In an analogous way, 
a {\it cube} is a word of the form $uuu$ for a nonempty word~$u$,
and the {\it period} of this cube is also the length of~$u$. A word
is called {\it square-free} ({\it cube-free}) if it contains no squares
(cubes) as factors. It is easy to see that there are no binary square-free 
words of length larger than~3. On the other hand, by the classical results
of Thue~\cite{Thue06,Thue12}, there exist ternary square-free words of 
arbitrary length and binary cube-free words of arbitrary length.
For ternary square-free words this result was strengthened
by Dejean in~\cite{Dejan}. She found ternary words of arbitrary length
which have no factors with exponents greater than $7/4$. On the other hand, 
she showed that any long enough ternary word contains a factor with an
exponent greater than or equal to $7/4$. Thus, the number $7/4$ is the
minimal limit for exponents of avoidable factors which is universally 
called {\it the repetition threshold} in arbitrarily long ternary words. 
Dejean conjectured also that the repetition threshold in arbitrarily long 
words over a $k$-letter alphabet is equal to $7/5$ for $k=4$ and $k/(k-1)$ for $k\ge 5$. 
This conjecture is now proved for any~$k$ through the work of several 
authors \cite{Carpi,CN1,CN2,CN3,Olag,Moh-Noor,Pans,Rao}.
%This conjecture was proved for $k=4$ by Pansiot~\cite{Pans}, for $5\le k\le 11$ 
%by Moulin Ollagnier~\cite{Olag}, for $12\le k\le 14$ by Mohammad-Noori
%and Currie~\cite{Moh-Noor}, and for $k\ge 38$ by Carpi~\cite{Carpi}. 

Denote the repetition threshold in arbitrarily long words over a $k$-letter 
alphabet by $\varphi_k$. In the paper we will call the words having no factors 
with exponents greater than $\varphi_k$ {\it minimally repetitive} words 
or {\it Dejean words}. By $S^{\langle k\rangle}(n)$ we denote the number 
of all minimally repetitive words of length~$n$ over a $k$-letter alphabet. 
Note that the set of all minimally repetitive words is obviously factorial. 
So for any~$k$ there exists the growth rate 
$\gamma^{\langle k\rangle}=\lim_{n\to\infty} \sqrt[n]{S^{\langle k\rangle}(n)}$.

The problem of estimating the number of repetition-free words has
been investigated actively during the last decades (reviews of
results on the estimations for the number of repetition-free words
obtained before 2008 can be found in~\cite{Berst, Grimm08}). 
The most progress in this field has been made for the case
of binary alphabet. In this case Dejean words reduce to overlap-free 
words which are also a classical object for combinatorial investigations.
It is proved in~\cite{RS} that the growth of the number of binary 
overlap-free words is polinomial. Actually, binary overlap-free words of each 
length are counted by a 2-regular function~\cite{Carpi0}. 

In~\cite{JIS} we proposed a new approach for obtaining lower bounds on 
the number of repetition-free words. Using this approach, we obtained 
precise lower bounds for the growth rates of ternary square-free words,
binary cube-free words, and ternary minimally repetitive words.
This approach proved to be very effective. In particular, in~\cite{Shur09}
Shur proposed an interesting modification of our approach which
allows to compute more effectively lower bounds for the growth rates
of words which contain no repetitions of exponent greater than or
equal to a given bound if this bound is not less than~$2$. 
The direction of our further investigations in this field is
testing the proposed approach for ``extreme'' cases when 
the prohibitions imposed on words are maximal possible for
the existence of words of arbitrary length avoiding these
prohibitions. These cases are obviously the most diffucult
for obtaining lower bounds on the number of appropriate words.
The case of minimally repetitive words is a natural example
of such ``extreme'' cases. Moreover, the general case of minimally 
repetitive words over a $k$-letter alphabet for $k\ge 5$ when
$\varphi_k=k/(k-1)$ is the most interesting for us. So this paper
is devoted to obtaining lower bounds on $\gamma^{\langle k\rangle}$
for $k\ge 5$ by using the proposed approach. Note that the method
proposed in~\cite{JIS} is not directly applicable to resolving
this problem because of the huge size of required computer computations. 
In this paper we propose an improvement of this method which requires
significantly fewer computer computations. Using this improvement,
we obtain lower bounds on $\gamma^{\langle k\rangle}$ for $5\le k \le 10$
which have the precision of $0.005$. As an evident consequence of these
results, we establish the exponential growth of the number of minimally
repetitive words over a $k$-letter alphabet for $5\le k \le 10$ 
(for $k=3, 4$ this fact was proved by Ochem in~\cite{OchemTIA}).

\section{Estimation for the number of minimally repetitive words}

\subsection{General}

For obtaining a lower bound on $\gamma^{\langle k\rangle}$ we 
will consider the alphabet $\Sigma_k=\{a_1,a_2,\ldots, a_k\}$ where $k\ge 5$.
We denote the set of all minimally repetitive words over~$\Sigma_k$ by~${\cal F}$.
By a prohibited factor we mean a factor with an exponent greater than $k/(k-1)$.
Let $m$ be a natural number, $m>k$, and $w',w''$ be two words from ${\cal F}(m)$.
We call the word $w''$ a {\it descendant} of the word $w'$ if 
$w'[2:m]=w''[1:m-1]$ and $w'w''[m]=w'[1]w''\in {\cal F}(m+1)$.
The word $w'$ is called in this case an {\it ancestor} of the 
word~$w''$. We introduce a notion of closed words in 
the following inductive way. A word $w$ from ${\cal F}(m)$ is 
called {\it right closed} ({\it left closed}) if and only if 
this word satisfies one of the two following conditions:

a) {\bf Basis of induction.} $w$ has no descendants (ancestors);

b) {\bf Inductive step.} All descendants (ancestors) of~$w$
are right closed (left closed).

\noindent
A word is {\it closed} if it is either right closed or left closed.
We denote by $\hat {\cal F}(m)$ the set of all words from ${\cal F}(m)$
which are not closed. By ${\cal L}_m$ we denote the set of all words
over $\Sigma_k$ such that the length of these words is not less than~$m$
and all factors of length~$m$ in these words belong to $\hat {\cal F}(m)$.
We also denote by ${\cal F}_m$ the set of all minimally repetitive words 
from ${\cal L}_m$. Note that a word $w$ is closed if and only if any 
word isomorphic to~$w$ is also closed. So we have the following
obvious fact.
\begin{proposition}
For any isomorphic words $w',w''$ and any $n\ge |w'|$ the 
equality $|{\cal F}_m^{(w')}(n)|=|{\cal F}_m^{(w'')}(n)|$ 
holds.
\label{utvizo}
\end{proposition}

A word will be called  {\it rarefied} if the distance between any two different
occurences of the same symbol in this word is not less than $k-1$.
\begin{proposition}
Any word from ${\cal L}_m$ is rarefied.
\label{rarefy}
\end{proposition} 

\begin{proof}
Let $w$ be an arbitrary word from ${\cal L}_m$. Assume
that $w[i]=w[j]$ where $j<i\le j+(k-2)$. Consider the factor $f=w[j:i]$.
Since $|f|=i-j+1\le k-1<m$, in~$w$ the factor $f$ is contained in some 
factor $f'$ of length~$m$. By the definiton of ${\cal L}_m$ we have
$f'\in {\cal F}(m)$, so $f\in {\cal F}$. On the other hand,
$f$ has the period $|f|-1$, so
$$
e(f)\ge \frac{|f|}{|f|-1}=\frac{i-j+1}{i-j}\ge \frac{k-1}{k-2} >
\frac{k}{k-1}
$$
which contradicts the definiton of ${\cal F}(m)$.
\end{proof}

A word $w$ of length $n\ge k-1$ over $\Sigma_k$ will be called {\it trimmed}
if $w[n-(k-1)+j]=a_j$ for $j=1,\ldots ,k-1$. We denote by $\hat {\cal F}'(m)$
the set of all trimmed words from $\hat {\cal F}(m)$. Taking into account
Proposition~\ref{rarefy}, it is not difficult to note that for any word 
from $\hat {\cal F}(m)$ there exists a single word from $\hat {\cal F}'(m)$ 
which is isomorphic to this word, and for any word from $\hat {\cal F}'(m)$ 
there exist exactly $k!$ different words from $\hat {\cal F}(m)$ which are 
isomorphic to this word. Thus $|\hat {\cal F}(m)|=k! |\hat {\cal F}'(m)|$.
Let $w',w''$ be two words from $\hat {\cal F}'(m)$. We call the word $w''$ 
a {\it quasi-descendant} of the word $w'$ if $w''$ is isomorphic to some 
descendant of~$w'$. The word $w'$ is called in this case a {\it quasi-ancestor} 
of the word~$w''$.

Let $\hat s=|\hat {\cal F}(m)|$ and $s=|\hat {\cal F}'(m)|$. Without loss of 
generality we can assume that $\hat {\cal F}(m)=\{w_1, w_2,\ldots, w_{\hat s}\}$
where $\hat {\cal F}'(m)=\{w_1, w_2,\ldots, w_s\}$. For any word $w$ from
$\hat {\cal F}(m)$ we will denote by $\iota(w)$ the serial number of~$w$ in
$\hat {\cal F}(m)$, i.e. $\iota(w)=i$ if $w=w_i$ for some $i=1, 2,\ldots,\hat s$.
We define a matrix $\hat \Delta_m=(\hat \delta_{ij})$ of size $\hat s\times\hat s$ 
in the following way: $\hat \delta_{ij}=1$ if and only if $w_i$ is an ancestor of $w_j$; 
otherwise $\hat \delta_{ij}=0$. For any natural~$t$ by $\hat \Delta_m^{(t)}=
(\hat \delta_{ij}^{(t)})$ we will denote the $t$-th power of the matrix $\hat \Delta_m$,
i.e.
$$
\hat \Delta_m^{(t)}=\underbrace{\hat \Delta_m\times\hat \Delta_m\times\ldots\times
\hat \Delta_m}_{t}.
$$
Further we use the following evident fact.
\begin{proposition}
For any $i, j=1, 2,\ldots,\hat s$ and any $n>m$ the equality
$|{\cal L}_m^{(w_i, w_j)}(n)|=\hat \delta_{ij}^{(n-m)}$ is valid.
\label{onLmuw}
\end{proposition} 

We also define a matrix $\Delta_m=(\delta_{ij})$ of size $s\times s$ 
in the following way: $\delta_{ij}=1$ if and only if $w_i$ is a
quasi-ancestor of $w_j$; otherwise $\delta_{ij}=0$. Note that $\Delta_m$ 
is a nonnegative matrix, so, by the Perron-Frobenius theorem, for $\Delta_m$ 
there exists some maximal in modulus eigenvalue~$r$ which is a nonnegative 
real number. Moreover, we can find some eigenvector $\tilde x=(x_1;\ldots; x_s)$ 
with nonnegative components which corresponds to~$r$. Assume that $r>1$ and 
all components of~$\tilde x$ are positive. Then we denote by $\mu$ the ratio 
$\max_{i=1,\ldots,s} x_i/\min_{i=1,\ldots,s} x_i$, and for $n\ge m$ we define 
$S_m^{\langle k\rangle}(n)=\sum_{i=1}^s x_i\cdot |{\cal F}_m^{(w_i)}(n)|$. 
In an inductive way we estimate $S_m^{\langle k\rangle}(n+1)$ by 
$S_m^{\langle k\rangle}(n)$.  

First we estimate $|{\cal F}_m^{(w)}(n+1)|$ for each $w\in \hat {\cal F}(m)$.
It is obvious that
\begin{equation}
|{\cal F}_m^{(w)}(n+1)|=|{\cal G}^{(w)}(n+1)|-|{\cal H}^{(w)}(n+1)|,
\label{Fwin}
\end{equation}
where ${\cal G}^{(w)}(n+1)$ is the set of all words~$v$ from 
${\cal L}_m^{(w)}(n+1)$ such that $v[1:n], v[n-m+1:n+1]\in {\cal F}$, 
and ${\cal H}^{(w)}(n+1)$ is the set of all words from ${\cal G}^{(w)}(n+1)$ 
which contain some prohibited factor as a suffix. 
If $w\in\hat {\cal F}'(m)$ we denote by $\pi (w)$ the set of all quasi-ancestors of~$w$. 
Taking into account Proposition~\ref{utvizo}, it is easy to see that
\begin{equation}
|{\cal G}^{(w)}(n+1)|=\sum_{v\in\pi (w)} |{\cal F}_m^{(v)}(n)|.
\label{Lwin}
\end{equation}
Therefore, using that $\tilde x$ is a eigenvector of $\Delta_m$ for
the eigenvalue~$r$, we obtain
\begin{eqnarray}
\nonumber
&\sum_{i=1}^s x_i\cdot |{\cal G}^{(w_i)}(n+1)|=
\sum_{i=1}^s \bigl( x_i\cdot \sum_{v\in\pi (w_i)} |{\cal F}_m^{(v)}(n)| \bigr)&\\
\nonumber
&=(x_1;x_2;\ldots; x_s)\left(\begin{array}{cccc}
\delta_{11}& \delta_{21}& \ldots & \delta_{s1}\\
\delta_{12}& \delta_{22}& \ldots & \delta_{s2}\\
\vdots     & \vdots     & \ddots & \vdots     \\
\delta_{1s}& \delta_{2s}& \ldots & \delta_{ss}
\end{array}\right)
\left(\begin{array}{c}
|{\cal F}_m^{(w_1)}(n)|\\
|{\cal F}_m^{(w_2)}(n)|\\
\vdots \\
|{\cal F}_m^{(w_s)}(n)|
\end{array}\right)&\\
\label{rSmn}
&=r\cdot (x_1;x_2;\ldots; x_s)
\left(\begin{array}{c}
|{\cal F}_m^{(w_1)}(n)|\\
|{\cal F}_m^{(w_2)}(n)|\\
\vdots \\
|{\cal F}_m^{(w_s)}(n)|
\end{array}\right)=r\cdot S_m^{\langle k\rangle}(n).&
\end{eqnarray}

We now estimate $|{\cal H}^{(w)}(n+1)|$. For any word~$v$ from
${\cal H}^{(w)}(n+1)$ we can find the minimal prohibited factor which is
a suffix of~$v$. We denote this factor by $h(v)$ and the minimal period 
of this factor by $\lambda (v)$. Since after removing the last symbol 
from $h(v)$ this factor can not be prohibited, we have actually 
$|h(v)|=\lfloor k\lambda (v)/(k-1)\rfloor +1$. Note that the 
value $\lambda (v)$ is not less than $p_0=(m+1)-\lfloor (m+1)/k \rfloor$. 
Thus
\begin{equation}
|{\cal H}^{(w)}(n+1)|=\sum_{j\ge p_0} |{\cal H}_j^{(w)}(n+1)|
\label{Lwin2}
\end{equation}
where ${\cal H}_j^{(w)}(n+1)$ is the set of all words~$v$ from ${\cal H}^{(w)}(n+1)$ 
such that $\lambda (v)=j$. 

\def\lenexp{t}
\def\tj#1{\left\lfloor\frac{k{#1}}{k-1}\right\rfloor +2}
\def\setjswjsw{X_{j,\lenexp}^{(w)}}
\def\setprefjswjsw{U_{j,\lenexp}^{(w)}}

\subsection{Upper bound for $|{\cal H}_j^{(w)}(n+1)|$}\label{ssecHj}

To estimate $|{\cal H}_j^{(w)}(n+1)|$, let $\chi(j)=\lfloor j/(k-1)\rfloor +1$ and let $\lenexp = j + \chi(j) + 1$.
Recall that for any~$v$ from ${\cal H}_j^{(w)}(n+1)$ the prohibited factor $h(v)$ 
is a word from ${\cal L}_m(j+\chi(j))$ with the minimal period~$j$.
Moreover, this word doesn't contain shorter prohibited factors and contains the word~$w$ 
as a suffix. 

Let $\setjswjsw$ be the set of words $v\in {\cal L}_m(\lenexp)$ such that $v[1:\lenexp-1]\in 
{\cal F}(\lenexp-1)$, $v[3:\lenexp]\in {\cal F}(\lenexp-2)$,
$v[\lenexp-j-\chi(j)+1:\lenexp-j]=v[\lenexp-\chi(j)+1:\lenexp]$ and $w$ is a suffix of $v$.
Note that for every $v\in \setjswjsw$, $v[\lenexp-j-\chi(j)]\ne v[\lenexp-\chi(j)]$, 
otherwise $v[1:\lenexp-1]$ would have a forbidden factor. 
Suppose that $n+1\ge \lenexp$ and let $u\in {\cal H}_j^{(w)}(n+1)$. Then
\begin{equation}
u[n'_j+1 : n+1]=u[n''_j+1 : n-j+1]
\label{forHji}
\end{equation}
where $n'_j=n-\lfloor j/(k-1)\rfloor$ and $n''_j=n-\lfloor kj/(k-1)\rfloor$. 
By definition of ${\cal H}_j^{(w)}(n+1)$, $u[n+1-\lenexp+1:n+1] \in \setjswjsw$. 
Moreover $u[1:n+1-\lenexp+m] \in {\cal F}(n+1-\lenexp+m)$. Thus we have
\begin{proposition}
$$|{\cal H}_j^{(w)}(n+1)| \le \sum_{v\in\setjswjsw} | {\cal F}^{(v[1:m])}(n+1-\lenexp+m) |.$$
\label{Hjiw}
\end{proposition}

Let $\setprefjswjsw$ be the multiset of all prefixes of size $m$ in words of $\setjswjsw$ 
(note that among words $\setprefjswjsw$ we can have identical words, i.e., the same word 
can be a prefix of different words of $\setjswjsw$ and so can be counted several times 
in $\setprefjswjsw$). Then Proposition~\ref{Hjiw} implies
$$
|{\cal H}_j^{(w)}(n+1)| \le \sum_{u\in \setprefjswjsw} |{\cal F}_m^{(u)}(n+1-\lenexp+m)|.
$$
For $l=1,\ldots ,s$, denote by $\zeta_{j,\lenexp}^{(l)}(w)$ the number of occurrences 
of $w_l$ in  the multiset $\setprefjswjsw$. Then
\begin{equation}
|{\cal H}_j^{(w)}(n+1)|\le \sum_{u\in \setprefjswjsw} |{\cal F}_m^{(u)}(n+1-\lenexp+m)|=
 \sum_{l=1}^s \zeta_{j,\lenexp}^{(l)}(w)\cdot |{\cal F}_m^{(w_l)}(n+1-\lenexp+m)|.
\label{Ljwin}
\end{equation}

\subsection{Weaker upper bound for $|{\cal H}_j^{(w)}(n+1)|$}\label{ssecHjw}

We can also obtain another estimation for $|{\cal H}_j^{(w)}(n+1)|$ where 
$w\in \hat {\cal F}'(m)$. This estimation is more rough in comparison with~(\ref{Ljwin})
but requires much fewer computer computations. To estimate $|{\cal H}_j^{(w)}(n+1)|$
by this way, we denote $\lfloor j/(k-1)\rfloor +1$ by $\chi(j)$ and assume that $\chi(j)\ge k-1$ 
and $n\ge j+m$. Recall that for any~$v$ from ${\cal H}_j^{(w)}(n+1)$ we have relation~(\ref{forHji}). 
We consider separately the two following cases: $\chi(j)\le m$ and $\chi(j)>m$.

Let $\chi(j)\le m$. For any~$v$ from ${\cal H}_j^{(w)}(n+1)$ denote by $f'(v)$ the factor $v[n+2-j-m : n+1-j]$ 
of~$v$. It follows from $v\in {\cal L}_m$ that $f'(v)\in \hat {\cal F}(m)$. Moreover, from~(\ref{forHji})
we obtain that $f'(v)$ and~$w$ have the common suffix of length $\chi(j)$. Since
$w\in \hat {\cal F}'(m)$ and $\chi(j)\ge k-1$, it implies that $f'(v)\in \hat {\cal F}'(m)$.
Thus
$$
|{\cal H}_j^{(w)}(n+1)|=\sum_{u\in W_j(w)} |{\cal I}_{j, u}^{(w)}(n+1)|
$$
where $W_j(w)$ is the set of all words from $\hat {\cal F}'(m)$ which have the common suffix 
of length $\chi(j)$ with the word~$w$, and ${\cal I}_{j, u}^{(w)}(n+1)$ is the set of all words~$v$
from ${\cal H}_j^{(w)}(n+1)$ such that $f'(v)=u$. To estimate $|{\cal I}_{j, u}^{(w)}(n+1)|$,
note that for any~$v$ from ${\cal I}_{j, u}^{(w)}(n+1)$ we have $v[1 : n+1-j]\in {\cal F}^{(u)}_m(n+1-j)$
and $v[n+2-j-m : n+1]\in {\cal L}_m^{(u, w)}(j+m)$. Hence, using Proposition~\ref{onLmuw}, we obtain
$$
|{\cal I}_{j, u}^{(w)}(n+1)|\le |{\cal F}^{(u)}_m(n+1-j)|\cdot |{\cal L}_m^{(u, w)}(j+m)|=
|{\cal F}^{(u)}_m(n+1-j)|\cdot \hat \delta_{\iota (u),\iota(w)}^{(j)}.
$$
Thus, in this case we get the estimation
\begin{equation}
|{\cal H}_j^{(w)}(n+1)|\le \sum_{u\in W_j(w)} \hat \delta_{\iota (u),\iota(w)}^{(j)}\cdot
|{\cal F}^{(u)}_m(n+1-j)|.
\label{onHjwn1}
\end{equation}

Let now $\chi(j)>m$. For any~$v$ from ${\cal H}_j^{(w)}(n+1)$ denote by $f''(v)$ the factor 
$v[n''_j+1 : n''_j+m]$ of~$v$. It follows from $v\in {\cal L}_m$ that $f''(v)\in \hat {\cal F}(m)$.
Thus in this case
$$
|{\cal H}_j^{(w)}(n+1)|=\sum_{u\in \hat {\cal F}(m)} |{\cal J}_{j, u}^{(w)}(n+1)|
$$
where ${\cal J}_{j, u}^{(w)}(n+1)$ is the set of all words~$v$ from ${\cal H}_j^{(w)}(n+1)$ 
such that $f''(v)=u$. To estimate $|{\cal J}_{j, u}^{(w)}(n+1)|$, consider an arbitrary
word~$v$ from ${\cal J}_{j, u}^{(w)}(n+1)$. Note that $v$ is determined uniquely by the
prefix $v[1 : n'_j+m]$ which satisfies the following conditions: $v[1 : n''_j+m]\in 
{\cal F}^{(u)}_m(n''_j+m)$, $v[n''_j+1 : n+1-j]\in {\cal L}_m^{(u, w)}(\chi(j))$, and
$v[n+2-j-m : n'_j+m]\in {\cal L}_m^{(w, u)}(j+2m-\chi(j))$. Hence, using Proposition~\ref{onLmuw}, 
we obtain
\begin{eqnarray*}
|{\cal J}_{j, u}^{(w)}(n+1)| &\le & |{\cal F}^{(u)}_m(n''_j+m)|\cdot |{\cal L}_m^{(u, w)}(\chi(j))|\cdot
|{\cal L}_m^{(w, u)}(j+2m-\chi(j))|\\
&= & |{\cal F}^{(u)}_m(n''_j+m)|\cdot \hat \delta_{\iota (u),\iota(w)}^{(\chi(j)-m)}
\cdot \hat \delta_{\iota (w),\iota(u)}^{(j+m-\chi(j))}.
\end{eqnarray*}
Thus, in this case we get the estimation
$$
|{\cal H}_j^{(w)}(n+1)|\le \sum_{u\in \hat {\cal F}(m)} \hat \delta_{\iota (u),\iota(w)}^{(\chi(j)-m)}\cdot
\hat \delta_{\iota (w),\iota(u)}^{(j+m-\chi(j))}\cdot |{\cal F}^{(u)}_m(n''_j+m)|.
$$
Taking into account Proposition~\ref{utvizo}, we can rewrite this estimation in the form
\begin{equation}
|{\cal H}_j^{(w)}(n+1)|\le \sum_{u\in \hat {\cal F}'(m)} |{\cal F}^{(u)}_m(n''_j+m)|
\left(\sum_{v\in {\cal K}(u)} \hat \delta_{\iota (v),\iota(w)}^{(\chi(j)-m)}\cdot
\hat \delta_{\iota (w),\iota(v)}^{(j+m-\chi(j))}\right).
\label{onHjwn2}
\end{equation}
Note that, unlike estimation~(\ref{Ljwin}), estimations (\ref{onHjwn1}) and~(\ref{onHjwn2}) can be computed
in polynomial time.

\subsection{Estimation of $|{\cal H}^{(w)}(n+1)|$}

We fix numbers $p_1, p_2$ such that $p_0\le p_1<p_2$ and $p_2\ge 2k-3$, and
assume for convenience that $n>kp_2/(k-1)$. We present sum~(\ref{Lwin2})
in the form
$$
\begin{array}{c}
\displaystyle
|{\cal H}^{(w)}(n+1)|=\sum_{j=p_0}^{p_1} |{\cal H}_j^{(w)}(n+1)|+
\sum_{j=p_1+1}^{p_2} |{\cal H}_j^{(w)}(n+1)|+|\hat {\cal H}^{(w)}(n+1)|
\end{array}
$$
where $\hat {\cal H}^{(w)}(n+1)=\bigcup_{j>p_2} {\cal H}_j^{(w)}(n+1)$.
Thus $\sum_{i=1}^s x_i |{\cal H}^{(w_i)}(n+1)|$ can be presented as
\begin{equation}
\sum_{j=p_0}^{p_1}\sum_{i=1}^s x_i |{\cal H}_j^{(w_i)}(n+1)| + 
\sum_{j=p_1+1}^{p_2}\sum_{i=1}^s x_i |{\cal H}_j^{(w_i)}(n+1)| 
+ \sum_{i=1}^s x_i |\hat {\cal H}^{(w_i)}(n+1)|.
\label{sumis}
\end{equation}

To estimate the first sum in~(\ref{sumis}), we use inequality~(\ref{Ljwin})
\begin{eqnarray}
\nonumber
\displaystyle
\sum_{j=p_0}^{p_1}\sum_{i=1}^s x_i |{\cal H}_j^{(w_i)}(n+1)| &\le &
\sum_{j=p_0}^{p_1}\sum_{i=1}^s x_i \sum_{l=1}^s \zeta_{j,\tj{j}}^{(l)}(w_i)\cdot 
|{\cal F}_m^{(w_l)}(n-\left\lfloor \frac{jk}{k-1} \right \rfloor -1 +m)\\ 
%\\(n+1-\tj{j}+m)|\\
\displaystyle
&=& \sum_{d=\tj{p_0}}^{\tj{p_1}} \sum_{l=1}^s \eta'_l(d)\cdot |{\cal F}_m^{(w_l)}(n-\left\lfloor \frac{jk}{k-1} \right\rfloor -1 +m)|
\label{sum_eta_j}
\end{eqnarray}
where $\eta'_l(d)=\sum_{i=1}^s x_i\cdot \zeta_{j,\tj{j}}^{(l)}(w_i)$ if there is a $j$ such that $\tj{j}=d$, and $\eta'_l(d)=0$ otherwise.

\medskip

To estimate the second sum in~(\ref{sumis}), we use inequalities (\ref{onHjwn1}) 
and~(\ref{onHjwn2}). In particular, in the case of $\chi(j)\le m$, using
inequality~(\ref{onHjwn1}) and taking into account that $u\in W_j(w)$ if and 
only if $w\in W_j(u)$, we obtain
\begin{eqnarray*}
\sum_{i=1}^s x_i |{\cal H}_j^{(w_i)}(n+1)| &\le & 
\sum_{i=1}^s \sum_{u\in W_j(w_i)} x_i \hat \delta_{\iota (u), i}^{(j)}\cdot
|{\cal F}^{(u)}_m(n+1-j)|\\
&=& \sum_{u\in \hat {\cal F}'(m)} |{\cal F}^{(u)}_m(n+1-j)|
\left( \sum_{w_i\in W_j(u)} x_i\cdot \hat \delta_{\iota (u), i}^{(j)}\right)\\
&=& \sum_{l=1}^s |{\cal F}^{(w_l)}_m(n+1-j)|
\left( \sum_{w_i\in W_j(w_l)} x_i\cdot \hat \delta_{l, i}^{(j)}\right).
\end{eqnarray*}
In the case of $\chi(j)>m$, using inequality~(\ref{onHjwn2}), we have
\begin{eqnarray*}
\sum_{i=1}^s x_i |{\cal H}_j^{(w_i)}(n+1)| &\le &
\sum_{i=1}^s \sum_{u\in \hat {\cal F}'(m)} x_i |{\cal F}^{(u)}_m(n''_j+m)|
\left(\sum_{v\in {\cal K}(u)} \theta^{(j)}_{i, v}\right)\\
&=& \sum_{u\in \hat {\cal F}'(m)} |{\cal F}^{(u)}_m(n''_j+m)|
\sum_{i=1}^s x_i \left(\sum_{v\in {\cal K}(u)} \theta^{(j)}_{i, v}\right)\\
&=& \sum_{l=1}^s |{\cal F}^{(w_l)}_m(n''_j+m)|
\sum_{i=1}^s x_i \left(\sum_{v\in {\cal K}(w_l)} \theta^{(j)}_{i, v}\right)
\end{eqnarray*}
where $\theta^{(j)}_{i, v}=\hat \delta_{\iota (v), i}^{(\chi(j)-m)}\cdot
\hat \delta_{i ,\iota(v)}^{(j+m-\chi(j))}$. Thus, defining $d(j)=j-1$ 
for the case of $\chi(j)\le m$ and $d(j)=\lfloor kj/(k-1)\rfloor -m$
for the case of $\chi(j)>m$, we conclude that
$$
\sum_{i=1}^s x_i |{\cal H}_j^{(w_i)}(n+1)|\le \sum_{l=1}^s 
\xi_l(j) \cdot |{\cal F}^{(w_l)}_m(n-d(j))|
$$
where
$$
\xi_l(j)=\left\{\begin{array}{rl}
\displaystyle
\sum_{w_i\in W_j(w_l)} x_i\cdot \hat \delta_{l, i}^{(j)},&\mbox{if $\chi(j)\le m$};\\
\displaystyle
\sum_{i=1}^s x_i \left(\sum_{v\in {\cal K}(w_l)} \theta^{(j)}_{i, v}\right),&\mbox{if $\chi(j)>m$}.\\
\end{array}
\right.
$$
Hence
$$
\sum_{j=p_1+1}^{p_2}\sum_{i=1}^s x_i |{\cal H}_j^{(w_i)}(n+1)|\le
\sum_{j=p_1+1}^{p_2}\sum_{l=1}^s \xi_l(j) \cdot |{\cal F}^{(w_l)}_m(n-d(j))|.
$$

We define $\xi'_l(d)=\xi_l(j)$ if there exists some~$j$ such that 
$d(j)=d$, and $\xi'_l(d)=0$ otherwise. Then 
\begin{equation}
\sum_{j=p_1+1}^{p_2}\sum_{l=1}^s \xi_l(j) \cdot |{\cal F}^{(w_l)}_m(n-d(j))|=
\sum_{d=d_2}^{d_3} \; \sum_{l=1}^s \xi'_l(d)\cdot |{\cal F}_m^{(w_l)}(n-d)|
\label{sum_xi_j}
\end{equation}
where $d_2=d(p_1+1)$, $d_3=d(p_2)$. 

\medskip

\def\stasum{a}
\def\endsum{b}

%Then we can present sum~(\ref{xxxx}) as
Summing up (\ref{sum_eta_j}) and (\ref{sum_xi_j}), we get
$$
\sum_{j=p_0}^{p_2}\sum_{i=1}^s x_i |{\cal H}_j^{(w_i)}(n+1)| \le \sum_{d=\stasum}^{\endsum} \; \sum_{l=1}^s \omega_l(d)\cdot |{\cal F}_m^{(w_l)}(n-d)|
$$
where $\omega_l(d)=\eta'_l(d)+\xi'_l(d)$, $\stasum=\min(d_2,\tj{p_0}-m-1)$ and $\endsum=\max(d_3,\tj{p_1}-m-1)$.
%where $d_0=\lfloor kp_0/(k-1)\rfloor -m+1$, $d_1=\lfloor kp_1/(k-1)\rfloor -m+1$.

We majorate this sum by some sum  $\sum_{d=\stasum}^{\endsum} \rho_d\cdot S_m^{\langle k\rangle}(n-d)$ 
in the following way. We compute consecutively coefficients $\rho_d$ of this sum for 
$d=\stasum, \stasum+1,\ldots ,\endsum$.  For each $d=\stasum, \stasum+1,\ldots ,\endsum-1$ together with 
the number $\rho_d$ we compute also numbers $\omega'_1(d+1),\ldots , \omega'_s(d+1)$ 
such that
\begin{equation}
\sum_{j=\stasum}^{d+1} \; \sum_{l=1}^s \omega_l(j)\cdot |{\cal F}_m^{(w_l)}(n-j)|
\le \sum_{l=1}^s \omega'_l(d+1)\cdot |{\cal F}_m^{(w_l)}(n-d-1)| +
\sum_{j=\stasum}^{d} \rho_j\cdot S_m^{\langle k\rangle}(n-j).
\label{sum_jd_1}
\end{equation}
For $d=\stasum$ we take $\rho_{\stasum}=\min_{1\le l\le s} (\omega_l(\stasum)/x_l)$. Then
$$
\sum_{l=1}^s \omega_l(\stasum)\cdot |{\cal F}_m^{(w_l)}(n-\stasum)|=
\rho_{\stasum}\cdot S_m^{\langle k\rangle}(n-\stasum) + \sum_{l=1}^s \nu_l\cdot 
|{\cal F}_m^{(w_l)}(n-\stasum)|
$$
where $\nu_l=\omega_l(\stasum)-\rho_{\stasum}\cdot x_l$, $l=1,\ldots, s$. Denote
by $\tilde\nu$ the vector $(\nu_1;\ldots ;\nu_s)$ and consider the vector
$\tilde\nu'=\Delta_m \tilde\nu$. Let $\tilde\nu'=(\nu'_1;\ldots ;\nu'_s)$.
It follows from (\ref{Fwin}) and~(\ref{Lwin}) that
$$
|{\cal F}_m^{(w_l)}(n-\stasum)|\le |{\cal G}^{(w_l)}(n-\stasum)|=
\sum_{v\in\pi (w_l)} |{\cal F}_m^{(v)}(n-\stasum-1)|
$$
for any $l=1,\ldots ,s$. Note also that $\nu_l\ge 0$ for $l=1,\ldots, s$.
Hence
\begin{eqnarray*}
\sum_{l=1}^s \nu_l\cdot |{\cal F}_m^{(w_l)}(n-\stasum)|&\le& \sum_{l=1}^s 
\left(\nu_l\cdot \sum_{v\in\pi (w_l)} |{\cal F}_m^{(v)}(n-\stasum-1)|\right)\\
&=&\sum_{l=1}^s \nu'_l\cdot |{\cal F}_m^{(w_l)}(n-\stasum-1)|.
\end{eqnarray*}
Thus
\begin{equation}
\sum_{j=\stasum}^{\stasum+1} \; \sum_{l=1}^s \omega_l(j)\cdot 
|{\cal F}_m^{(w_l)}(n-j)|\le\rho_{\stasum}\cdot 
S_m^{\langle k\rangle}(n-\stasum)+\sum_{l=1}^s \omega'_l(\stasum+1)\cdot 
|{\cal F}_m^{(w_l)}(n-\stasum-1)|
\label{sum_j\stasum}
\end{equation}
where $\omega'_l(\stasum+1)=\omega_l(\stasum+1)+\nu'_l$. Assume now that for
some~$d$ such that $\stasum<d<\endsum$ we already computed the numbers $\rho_{\stasum},
\ldots ,\rho_{d-1}$ and $\omega'_1(d),\ldots,\omega'_s(d)$. Then we take
$\rho_{d}=\min_{1\le l\le s} (\omega'_l(d)/x_l)$, $\tilde\nu=(\omega'_1(d)-
\rho_d\cdot x_1,\ldots ,\omega'_s(d)-\rho_d\cdot x_s)$, and $\tilde\nu'=
\Delta_m \tilde\nu$. We take also $\omega'_l(d+1)=\omega_l(d+1)+\nu'_l$ 
where $\nu'_l$ is the $l$-th component of the vector $\tilde\nu'$, 
$l=1,\ldots ,s$. Analogously to inequality~(\ref{sum_j\stasum}),
in this case we have the inequality
$$
\begin{array}{c}
\displaystyle
\sum_{l=1}^s \left(\omega'_l(d)\cdot |{\cal F}_m^{(w_l)}(n-d)|+
\omega_l(d+1)\cdot |{\cal F}_m^{(w_l)}(n-d-1)|\right)\\
\displaystyle
\le \rho_d\cdot S_m^{\langle k\rangle}(n-d)+\sum_{l=1}^s \omega'_l(d+1)\cdot 
|{\cal F}_m^{(w_l)}(n-d-1)|.
\end{array}
$$
This inequality implies that inequality~(\ref{sum_jd_1})
holds for every~$d$. For $d=\endsum$ we take $\rho_{\endsum}=\max_{1\le l\le s} 
(\omega'_l(\endsum)/x_l)$. Thus,
$$
\sum_{d=\stasum}^{\endsum} \; \sum_{l=1}^s \omega_l(d)\cdot |{\cal F}_m^{(w_l)}(n-d)|\le
\sum_{d=\stasum}^{\endsum} \rho_d\cdot S_m^{\langle k\rangle}(n-d)
$$
which implies
\begin{equation}
\sum_{j=p_0}^{p_2}\sum_{i=1}^s x_i |{\cal H}_j^{(w_i)}(n+1)| \le
\sum_{d=\stasum}^{\endsum} \rho_d\cdot S_m^{\langle k\rangle}(n-d).
\label{sum_i_s1}
\end{equation}

\subsection{Upper bound for $|\hat {\cal H}^{(w_i)}(n+1)|$}

We estimate finally the sum $\sum_{i=1}^s x_i |\hat {\cal H}^{(w_i)}(n+1)|$.
For this purpose we denote by $\hat {\cal H}(n+1)$ the set $\bigcup_{i=1}^{\hat s}\hat {\cal H}^{(w_i)}(n+1)$
and by $\hat {\cal H}'(n+1)$ the set $\bigcup_{i=1}^s\hat {\cal H}^{(w_i)}(n+1)$.
Note that the sets $\hat {\cal H}^{(w_i)}(n+1)$ are non-overlapping, so
$|\hat {\cal H}'(n+1)|=\sum_{i=1}^s |\hat {\cal H}^{(w_i)}(n+1)|$. Thus 
\begin{equation}
\sum_{i=1}^s x_i |\hat {\cal H}^{(w_i)}(n+1)|\le |\hat {\cal H}'(n+1)|\cdot\max_{i=1,\ldots,s} x_i.
\label{onhatH}
\end{equation}
Moreover, since by Proposition~\ref{rarefy} any word from $\hat {\cal H}(n+1)$ is rarified and $n+1>k-1$,
for any word from $\hat {\cal H}(n+1)$ there exists a single word from $\hat {\cal H}'(n+1)$ 
which is isomorphic to this word, and for any word from $\hat {\cal H}'(n+1)$ there exist 
exactly $k!$ different words from $\hat {\cal H}(n+1)$ which are isomorphic to this word.
So $|\hat {\cal H}(n+1)|=k!|\hat {\cal H}'(n+1)|$.

Let $v$ be an arbitrary word from $\hat {\cal H}(n+1)$. Then for~$v$
we have
$$
v[n-\left\lfloor \frac{\lambda (v)}{k-1}\right\rfloor +1 : n+1]=v[n'-\left\lfloor \frac{\lambda (v)}{k-1}\right\rfloor +1 : n'+1]
$$
where $n'=n-\lambda (v)$. Thus the word~$v$ is determined uniquely by the number 
$\lambda (v)$ and the prefix $v[1:n-\lfloor \lambda (v)/(k-1)\rfloor ]$. 
We denote this prefix by $\tau (v)$. Further we use the following fact.
\begin{lemma}
For any different $v', v''\in \hat {\cal H}(n+1)$ the prefixes $\tau (v'), \tau (v'')$
are also different.
\label{mainlemma}
\end{lemma}

\begin{proof}
Let $\tau (v')=\tau (v'')=u$ for some different $v', v''\in \hat {\cal H}(n+1)$.
Denote by~$l$ the length of~$u$. Note that $v', v''\in {\cal L}_m$, so $v'$, $v''$ and $u$
are rarefied by Proposition~\ref{rarefy}. Thus without loss of generality we can assume that 
$u$ is trimmed, i.e.
\begin{equation}
a_j=u[l-(k-1)+j]=v'[l-(k-1)+j]=v''[l-(k-1)+j]
\label{assumpt}
\end{equation} 
for $j=1,\ldots ,k-1$. As we noted above, the equalities $\tau (v')=\tau (v'')$
and $\lambda (v')=\lambda (v'')$ imply $v'=v''$. So $\lambda (v')\neq\lambda (v'')$.
Without loss of generality we assume that $\lambda (v')>\lambda (v'')$.
Since $n-l=\lfloor\lambda (v')/(k-1)\rfloor=\lfloor\lambda (v'')/(k-1)\rfloor$, 
we can assume moreover that $\lambda (v'')<\lambda (v')<\lambda (v'')+(k-1)$.
Note also that the inequality $\lfloor\lambda (v')/(k-1)\rfloor\ge 2$ follows
from $\lambda (v')\ge p_2+1\ge 2k-2$. So $l=n-\lfloor\lambda (v')/(k-1)\rfloor\le n-2$.
Recall that we have also
\begin{equation}
v'[l+1:n+1]=v'[l-\lambda (v')+1:n-\lambda (v')+1]=u[l-\lambda (v')+1:n-\lambda (v')+1],
\label{forv1}
\end{equation}
\begin{equation}
v''[l+1:n+1]=v''[l-\lambda (v'')+1:n-\lambda (v'')+1]=u[l-\lambda (v'')+1:n-\lambda (v'')+1].
\label{forv2}
\end{equation}

Suppose $v'[l+1]=v''[l+1]$. Then by equations (\ref{forv1}) and~(\ref{forv2}) we obtain
$u[l-\lambda (v')+1]=u[l-\lambda (v'')+1]$. Since 
$$
(l-\lambda (v'')+1)-(l-\lambda (v')+1)=\lambda (v')-\lambda (v'')\le k-2,
$$
this contradicts that $u$ is rarefied. So $v'[l+1]\neq v''[l+1]$. Since $v'$, $v''$
are rarefied, it is easy to note from~(\ref{assumpt}) that $v'[l+1]$ and $v''[l+1]$ 
can be either $a_1$ or $a_k$. So we have only two possible cases: $v'[l+1]=a_1$, $v''[l+1]=a_k$
or $v'[l+1]=a_k$, $v''[l+1]=a_1$. We consider these cases separately.

Let $v'[l+1]=a_1$ and $v''[l+1]=a_k$. Then it is easy to note that the symbol $v'[l+2]$ can be
only $a_k$. Thus, by equations (\ref{forv1}) and~(\ref{forv2}) we obtain
$u[l-\lambda (v')+1]=a_1$, $u[l-\lambda (v')+2]=a_k$ and $u[l-\lambda (v'')+1]=a_k$.
So $u[l-\lambda (v')+2]=u[l-\lambda (v'')+1]$. Since
$$
(l-\lambda (v'')+1)-(l-\lambda (v')+2)=\lambda (v')-\lambda (v'')-1<k-1
$$
and $u$ is rarefied, the only case we have to consider is $l-\lambda (v'')+1=l-\lambda (v')+2$,
i.e. $\lambda (v')-\lambda (v'')=1$ (in this case $u[l-\lambda (v')+2]$ and $u[l-\lambda (v'')+1]$ 
are the same letter in~$u$). Since $v''$ is rarefied, $v''[l+2]$ can be either $a_1$ or $a_2$.
If $v''[l+2]=a_1$, then by~(\ref{forv2}) we obtain $a_1=u[l-\lambda (v'')+2]=u[l-\lambda (v')+3]$. 
Thus we have in this case that $u[l-\lambda (v')+1]=u[l-\lambda (v')+3]$ which contradicts that 
$u$ is rarefied since $2<k-1$. Let $v''[l+2]=a_2$. Then it is easy to note that the symbol $v''[l+3]$ 
can be only $a_1$. Therefore, $a_1=u[l-\lambda (v'')+3]=u[l-\lambda (v')+4]$ by~(\ref{forv2}).
Thus we have that $u[l-\lambda (v')+1]=u[l-\lambda (v')+4]$ which contradicts again that 
$u$ is rarefied.

Let now $v'[l+1]=a_k$ and $v''[l+1]=a_1$. Then it is easy to note that the symbol $v''[l+2]$ 
can be only $a_k$. Thus, by equations (\ref{forv1}) and~(\ref{forv2}) we obtain
$u[l-\lambda (v')+1]=a_k$, $u[l-\lambda (v'')+1]=a_1$ and $u[l-\lambda (v'')+2]=a_k$. 
Since $u$ is rarefied, we have
$$
(l-\lambda (v'')+2)-(l-\lambda (v')+1)=\lambda (v')-\lambda (v'')+1\ge k-1.
$$
Thus $\lambda (v')-\lambda (v'')=k-2$ has to be valid in this case. Since $v'$ is rarefied,
we have also that $v'[l+2]$ can be either $a_1$ or $a_2$. If $v'[l+2]=a_1$, then 
$u[l-\lambda (v')+2]=a_1$ by~(\ref{forv1}). Since $u[l-\lambda (v'')+1]=a_1$ and
$$
(l-\lambda (v'')+1)-(l-\lambda (v')+2)=\lambda (v')-\lambda (v'')-1=k-3<k-1,
$$
this contradicts that $u$ is rarefied. Let $v'[l+2]=a_2$. It is easy to note that in
this case the symbol $v'[l+3]$ can be only $a_1$. Therefore, $u[l-\lambda (v')+3]=a_1$ 
by~(\ref{forv1}). Taking into account that $u[l-\lambda (v'')+1]=a_1$ and $k\ge 5$, we
obtain again a contradiction with the fact that $u$ is rarefied, so the lemma is proved. 
\end{proof}

Note that for any word $v\in \hat {\cal H}(n+1)$ we have $\tau (v)\in {\cal F}_m$ and
$n-\lfloor n/k\rfloor\le |\tau (v)|\le n-\lfloor (p_2+1)/(k-1)\rfloor$, i.e.
$\tau (v)\in {\cal Q}(n+1)=\bigcup_{j=n-\lfloor n/k\rfloor}^{n-\lfloor (p_2+1)/(k-1)\rfloor} {\cal F}_m(j)$.
So from Lemma~\ref{mainlemma} we obtain that $|{\cal Q}(n+1)|\ge |\hat {\cal H}(n+1)|=
k!|\hat {\cal H}'(n+1)|$. Denote by ${\cal Q}'(n+1)$ the set of all trimmed words
from ${\cal Q}(n+1)$. Since by Proposition~\ref{rarefy} any word from ${\cal Q}(n+1)$ 
is rarified and has the length greater than $p_2>k-1$, for any word from ${\cal Q}(n+1)$ 
there exists a single word from ${\cal Q}'(n+1)$ which is isomorphic to this word, 
and for any word from ${\cal Q}'(n+1)$ there exist exactly $k!$ different words 
from ${\cal Q}(n+1)$ which are isomorphic to this word. So 
$|{\cal Q}(n+1)|=k!|{\cal Q}'(n+1)|$. Thus $|{\cal Q}'(n+1)|\ge |\hat {\cal H}'(n+1)|$.
Note that actually ${\cal Q}'(n+1)=\bigcup_{j=n-\lfloor n/k\rfloor}^{n-\lfloor (p_2+1)/(k-1)\rfloor}
\bigcup_{i=1}^s {\cal F}_m^{(w_i)}(j)$ and, since all sets ${\cal F}_m^{(w_i)}(j)$
are non-overlapping, 
$$
|{\cal Q}'(n+1)|=\sum_{j=n-\left\lfloor \frac{n}{k}\right\rfloor}^{n-\left\lfloor \frac{p_2+1}{k-1}\right\rfloor}
\sum_{i=1}^s |{\cal F}_m^{(w_i)}(j)|\le \sum_{j=n-\left\lfloor \frac{n}{k} \right\rfloor}^{n-\left\lfloor \frac{p_2+1}{k-1}\right\rfloor} {S_m^{\langle k\rangle}(j)}/{(\min_{i=1,\ldots,s} x_i)}.
$$
Thus, taking into account~(\ref{onhatH}), we obtain

\begin{equation}
\begin{array}{c}
\displaystyle
\sum_{i=1}^s x_i |\hat {\cal H}^{(w_i)}(n+1)|\le |\hat {\cal H}'(n+1)|\cdot
\max_{i=1,\ldots,s} x_i\le |{\cal Q}'(n+1)|\cdot\max_{i=1,\ldots,s} x_i\\
\displaystyle
\le (\max_{i=1,\ldots,s} x_i)\sum_{j=n-\left\lfloor \frac{n}{k}\right\rfloor}^{n-\left\lfloor \frac{p_2+1}{k-1}\right\rfloor}
{S_m^{\langle k\rangle}(j)}/{(\min_{i=1,\ldots,s} x_i)}=
\mu \sum^{\left\lfloor \frac{n}{k}\right\rfloor}_{d=\left\lfloor \frac{p_2+1}{k-1}\right\rfloor}S_m^{\langle k\rangle}(n-d).
\end{array}
\label{sum_i_send}
\end{equation}

\subsection{Getting a lower bound for $\gamma^{\langle k\rangle}$}

Summing up estimation (\ref{sum_i_send}) with relation (\ref{sum_i_s1}), we conclude that
\begin{equation}
\begin{array}{c}
\displaystyle
\sum_{i=1}^s x_i |{\cal H}^{(w_i)}(n+1)|\le \sum_{d=\stasum}^{\endsum} \rho_d\cdot 
S_m^{\langle k\rangle}(n-d)
\displaystyle
+\mu \sum^{\left\lfloor \frac{n}{k}\right\rfloor}_{d=\left\lfloor \frac{p_2+1}{k-1}\right\rfloor} S_m^{\langle k\rangle}(n-d).
\end{array}
\label{sum_i_all}
\end{equation}

For the sake of convenience we denote by ${\cal P}(z)$ the 
polynomial $\sum_{d=\stasum}^{\endsum} \rho_d\cdot z^d$ in a variable~$z$.
Suppose for some $\alpha >1$ we have
\begin{equation}
S^{\langle k\rangle}_m(n) \ge \alpha^d \cdot S^{\langle k\rangle}_m(n-d) 
\label{pri}
\end{equation}
for each $d=1, 2,\ldots , n-m$. Then relation~(\ref{sum_i_all}) implies that 
\begin{eqnarray*}
\sum_{i=1}^s x_i |\hat {\cal H}^{(w_i)}(n+1)| &\le & 
S^{\langle k\rangle}_m(n) \sum_{d=\stasum}^{\endsum}\frac{\rho_d}{\alpha^d}
+\mu S^{\langle k\rangle}_m(n) 
\sum^{\left\lfloor \frac{n}{k}\right\rfloor}_{d=\left\lfloor \frac{p_2+1}{k-1}\right\rfloor}\frac{1}{\alpha^d}\\
&<& S^{\langle k\rangle}_m(n)\left({\cal P}(\frac{1}{\alpha})
+\mu\sum^{\infty}_{d=\left\lfloor \frac{p_2+1}{k-1}\right\rfloor}
\frac{1}{\alpha^d}\right)\\
&=& S^{\langle k\rangle}_m(n)
\left({\cal P}(\frac{1}{\alpha})+\frac{\mu}{\alpha^q (\alpha -1)}
\right)
\end{eqnarray*}
where $q=\left\lfloor \frac{p_2+1}{k-1}\right\rfloor -1$. Using this estimation and equalities (\ref{Fwin}) 
and~(\ref{rSmn}), we obtain
$$
\begin{array}{c}
\displaystyle
S^{\langle k\rangle}_m(n+1)=\sum_{i=1}^s x_i\cdot |{\cal G}^{(w_i)}(n+1)|-
\sum_{i=1}^s x_i |\hat {\cal H}^{(w_i)}(n+1)|\\
\displaystyle
>S^{\langle k\rangle}_m(n)\cdot
\left(r-{\cal P}(\frac{1}{\alpha})-\frac{\mu}{\alpha^q (\alpha -1)}\right).
\end{array}
$$
Therefore, if $\alpha$ satisfy the inequality
$$
r-{\cal P}(\frac{1}{\alpha})-\frac{\mu}{\alpha^q (\alpha -1)}
\ge\alpha,
$$
we obtain the inequality $S^{\langle k\rangle}_m({n+1})\ge \alpha S^{\langle k\rangle}_m(n)$, and thus
$S^{\langle k\rangle}_m(n+1) \ge \alpha^d \cdot S^{\langle k\rangle}_m(n+1-d)$ holds for any $d=1,2,\ldots, n-m+1$.
If inequalities~(\ref{pri}) hold for some $n'$, then inequalities~(\ref{pri}) 
hold inductively in this case for every $n\ge n'$.
Thus we have $S^{\langle k\rangle}_m(n)=\Omega (\alpha^n)$. 
Since, obviously, the order of growth 
of $S^{\langle k\rangle}(n)$ is not less than $S^{\langle k\rangle}_m(n)$, we then 
conclude that $S^{\langle k\rangle}(n)=\Omega (\alpha^n)$. Hence 
$\gamma^{\langle k\rangle}\ge\alpha$.

Note that for obtaining the bound $\gamma^{\langle k\rangle}\ge\alpha$ we have to prove initially that inequalities~(\ref{pri}) holds for $n'$. 
For these purposes we compute the exact values of $S^{\langle k\rangle}(n)$ for $n\le n_0$ by an enumeration of all Dejean's words of size at most $n_0$. 
The inequalities $S^{\langle k\rangle}_m({n+1})\ge \alpha S^{\langle k\rangle}_m(n)$ for $n_0<n\le kp_2/(k-1)$ could be verified 
in the same inductive way as described above with evident modifications following from the restriction $n\le kp_2/(k-1)$.

\section{Results}

Using the described method of estimating $\gamma^{\langle k\rangle}$, we obtained lower
bounds on $\gamma^{\langle k\rangle}$ for $5\le k \le 10$. The obtained bounds together 
with the parameters $m$, $n_0$, $p_1$, $p_2$ used in the computer computations of these bounds
are given in the following table. In this table we give also the upper bounds on 
$\gamma^{\langle k\rangle}$ we obtain with the method described in~\cite{Shur08}.
For the anti-dictionary ${\cal A}$, we take the set of all binary minimally forbidden 
words in the Pansiot's code (w.r.t. factor containment) of size at most~$q$.

\begin{center}
\small
\begin{tabular}{|c||c|c|c|c|c|c||c|c|c|}
\hline
$k$ & $m$ & $s$ & $n_0$ & $p_1$& $p_2$
& lower bound on $\gamma^{\langle k\rangle}$
& $q$
& ${\vert \cal A \vert}$
& upper bound on $\gamma^{\langle k\rangle}$\\
%& upp. bnd. in \cite{Shur08}\\ 
\hline
\hline
%5 12 : Aff marche pas
%5& 35& 551 &  & >160... & 900 & (phenom) & ... &1.158151\\
5&50 & 5287 & 150 & 183 & 600 & 1.153811 & 158 & 12783585 & 1.157895\\
% &1.158151\\
%inf todo phenom2
%
%5& 56& 12866 & 170& & 350& 1.154717& & & &1.158151\\
%
%up 1.15801451009 (134)
%up 1.1579692132374365319 (141)
%up 1.1579419069309 (146)
%up 1.157928654864401792733 (148)
%up 1.1579190683 (151)
%up 1.1579135946 (152)
%up 1.15790729975 (153) size:8645036
%up 1.15789795065496634858 (156 hagrid) size:11333568
%up 1.15789435050825 157 size:12783585
\hline
%6 14 : Aff marche pas
%6& 33& 1926& 120& & 270& 1.220299& 1.224837\\
%6 40 ??
6& 33& 1926& 100 & 125& 500& 1.223437 & 113 & 3946990 & 1.224695\\
% & 1.224837\\
%
%up 1.22470221544056198 (111)
%up 1.22469428727859057 (112)
\hline
%7 16: Aff devrait pas marcher
7 & 28 & 318 & 100 & 126 & 600 & 1.236409 & 114 & 2958045 & 1.236899\\
% & 1.236964 \\
%OK hagrid
%up  1.23689877863405885 (113 sur calcul1)
\hline
8& 18& 31 & 100 & 119 & 600& 1.234725  & 118 & 1399465 &  1.234843\\ 
%& 1.234864\\
%OK hagrid
%
%%8& 25& 139& 117& & 297& 1.229356& 1.234864\\
%8& 25& 139& & & 600& pas utile ? (meme eigen que 18) & 1.234864\\
%
%up  1.234844171261625587544 (113) 
%up  1.234843083448442303 (116 calcul1)
%up  1.234842725736074240759047 (117 calcul1)
\hline
9& 20& 42 & 100 & 123 & 600 & 1.246659  & 112 & 287646 & 1.246678\\ 
%&1.246685\\ 
%OK hagrid
%
%up  1.24667772889999301972 (111 phenom) size:287646 STOP
\hline
10& 22& 55 & 100 & 122 & 600 & 1.239287  & 115 & 65346 & 1.239308\\ 
%& 1.239310\\
%OK hagrid
%
%up  1.2393074885087728870767 (114 calcul1)
\hline
\end{tabular}
\end{center}

Comparing the obtained lower bounds with the  upper bounds on $\gamma^{\langle k\rangle}$ presented in the table, one can conclude 
that we have estimated $\gamma^{\langle\mbox{k}\rangle}$ for $5\le k\le 10$
with the precision of $0.005$.

\section{Conclusion}

In this paper we obtained lower bounds on $\gamma^{\langle k\rangle}$ for $5\le k\le 10$, 
but we believe that by the method proposed in the paper lower bounds on $\gamma^{\langle k\rangle}$
could be computed for any fixed $k\ge 5$ (provided that $\gamma^{\langle k\rangle}>1$).
So we consider as an interesting problem for further investigations the question if
the computations described in the paper can be generalized theoretically for obtaining
theoretical lower bounds on $\gamma^{\langle k\rangle}$ valid for any $k\ge 5$.

\section*{Acknowledgments}

This work started when both authors were invited to LIAFA, University
Paris Diderot (Paris-7), France, in June 2009. R.Kolpakov acknowledges the partial
support of the Russian Foundation for Fundamental Research (Grant 08--01--00863) 
and of the program for supporting Russian scientific schools (Grant NSh 5400.2006.1).

\end{document}